\documentclass[english]{llncs}
\usepackage[latin9]{inputenc}
\usepackage{babel}

\usepackage{tipa}
\usepackage{amsmath}
\usepackage{caption}
\usepackage{subcaption}
\usepackage{graphicx}
\usepackage{amssymb}
\usepackage{color}
\usepackage{wrapfig}
\usepackage[unicode=true, pdfusetitle,
 bookmarks=true,bookmarksnumbered=false,bookmarksopen=false,
 breaklinks=false,pdfborder={0 0 1},backref=false,colorlinks=false]
 {hyperref}

\makeatletter

\newtheorem{thm}{Theorem}
\newtheorem{defn}[thm]{Definition}
\newtheorem{lem}[thm]{Lemma}

\newtheorem{cor}[thm]{Corollary}

\usepackage[nothing]{algorithm} \usepackage[noend]{algorithmic} 
\usepackage{float}
\newdimen\algorithmicindent \algorithmicindent=0.5cm

\usepackage{color}
\graphicspath{{figures/}{./}}

\def\eps{\varepsilon}

\newif\ifnotesw\noteswtrue

\newif\ifnotesw\noteswtrue

\def\blfootnote{\xdef\@thefnmark{}\@footnotetext}

\makeatother

\usepackage[affil-it]{authblk}

\begin{document}

\title{Network installation and recovery: approximation lower bounds and faster exact formulations}
\date{}

\author{Alexander Gutfraind\inst{1} \and
Jeremy Kun\inst{2} \and
\'Ad\'am D. Lelkes\inst{2} \and
Lev Reyzin\inst{2}}  

\institute{
Division of Epidemiology and Biostatistics\\
\and
Department of Mathematics, Statistics, and Computer Science\\
University of Illinois at Chicago\\
\texttt{\{agutfrai,jkun2,alelke2,lreyzin\}@uic.edu}
}

\maketitle

\begin{abstract}
We study the Neighbor Aided Network Installation Problem (NANIP) introduced
previously which asks for a minimal cost ordering of the vertices of
a graph, where the cost of visiting a node is a function of the number of
neighbors that have already been visited. This problem has applications in
resource management and disaster recovery. In this paper we analyze the
computational hardness of NANIP. In particular we show that this problem is
NP-hard even when restricted to convex decreasing cost functions, give a
linear approximation lower bound for the greedy algorithm, and prove a  
general sub-constant approximation lower bound. Then we give a new
integer programming formulation of NANIP and empirically observe its speedup
over the original integer program.

\end{abstract}

\section{Introduction}
\setcounter{footnote}{3} 

We motivate our study with an example from infrastructure networks.  It is
well known that many vital infrastructure systems can be represented as
networks, including transport, communication and power networks.  Large parts
of these networks can be severely damaged in the event of a natural disaster.
When faced with large-scale damage, authorities must develop a plan for
restoring the networks. A particularly challenging aspect of the recovery is
the lack of infrastructure, such as roads or power, necessary to support the
recovery operations.  For example, to clear and rebuild roads, equipment must
be brought in, but many of the access roads are themselves blocked and damaged.
Abstractly, as the recovery progresses, previously recovered nodes provide
resources that help reduce the cost of rebuilding their neighbors. We call this
phenomenon ``neighbor aid''.

Recently, \cite{Gutfraind14} introduced and analyzed a simple
model of neighbor aided recovery in terms of a convex discrete optimization
problem called the \emph{Neighbor Aided Network Installation Problem} (NANIP).
We will henceforth use the terms ``recover'' and ``install'' interchangeably.
For simplicity, we assume that during the recovery of a network all of its
nodes and edges must be visited and restored. They asked how to optimize the
recovery schedule in order to minimize the total cost?  This is also the
question we address herein.

In the NANIP model, the cost of recovering a node depends only on the number of
its already recovered neighbors, capturing the intuition that neighbor aid is
the determining factor of the cost of rebuilding a new node.  NANIP offers a
stylized model for disaster recovery of networks (among other applications) but
the interest in disaster recovery of networks is not new.  A partial list of
existing studies include
\cite{Guha99,nurre2010restoring,Lee07,Adibi94,Bertoli02,coffrin2011strategic}.
A common framework is to consider infrastructure systems as a set of
interdependent network flows, and formulate the problem of minimizing the cost
of repairing such damaged networks.  Another class of models
\cite{Hentenryck10} develops a stochastic optimization problem for stockpiling
resources and then distributing them following a disaster.  More abstract
problems related to NANIP are the single processor scheduling problem
\cite{Karp61}, the linear ordering problem \cite{Mitchell96}, and the study of
tournaments in graph theory \cite{West01}.  

NANIP assumes that certain tasks are dependent and cannot be performed in
parallel, but unlike many scheduling problems, there are no partial order
constraints.  Similarly to traveling salesman problem (TSP)
\cite{schrijver2005history}, the NANIP problem also asks for an optimal
permutation of the vertices of the graph but, unlike in the case of the
traveling salesman problem, the cost associated with visiting a given node
could depend on \emph{all} of the nodes visited before the given node. Another
key difference between NANIP and TSP is that in NANIP it is allowed to visit
nodes that are not neighbors of any previously-visited nodes. As we will see,
such disconnected traversals provide $\Omega(\log(n))$ multiplicative
improvements over connected ones.

Since neighbor aid is assumed to reduce the cost of recovery, we are mainly
interested in decreasing cost functions. Furthermore, since convexity for
decreasing functions captures the ``law of diminishing returns'', i.e. that as
the number of recovered neighbors increases, the per-node value of the aid
provided by one neighbor decreases, convex decreasing functions are of special
interest. Although \cite{Gutfraind14}  gave NP-hardness of
NANIP for general cost via a straightforward reduction from Maximum Independent
Set, the cost function used there was increasing, thus leaving the complexity
of the convex decreasing case an open question.  In this paper we show this
problem is NP-hard as well.  We also provide a new convex integer programming
formulation and analyze the performance of the greedy algorithm, showing that
its worst case approximation ratio is $\Theta(n)$.

\section{Preliminaries}
An instance of NANIP is specified by an undirected graph $G=(V,E)$ and a
real-valued function $f: \mathbb{N} \to \mathbb{R}_{\geq 0}$. The function $f$
represents the cost of installing a vertex $v$, where the argument is the
number of neighbors of $v$ that have already been installed. Hence, the domain
of $f$ is the non-negative integers, bounded by the maximum degree of $G$ (for
terminology see \cite{West01}).  The goal is to find a permutation of the nodes
that minimizes the total cost of the network installation. The cost of
installing node $v_t \in V$ under a permutation $\sigma$ of $V$ is given by
$$f(r(v_t, G, \sigma))\,,$$ where $r(v_t, G, \sigma)$ is the number of nodes
adjacent to $v_t$ in $G$ that appear before $v_t$ in the permutation $\sigma$.
The total cost of installing $G$ according to $\sigma$ is given by
\begin{equation} C_G(\sigma) = \sum_{t=1}^{n} f(r(v_t, G, \sigma)).
\label{eq:general-NANIP} \end{equation}

The problem is illustrated in Fig.~\ref{fig:illustration}.  Generally, the
choice of $f$ depends on the application, and $f$ will often be convex
decreasing. 

\begin{figure}[th]
\begin{subfigure}[b]{0.4\textwidth}
\centering
\includegraphics[width=0.5\textwidth]{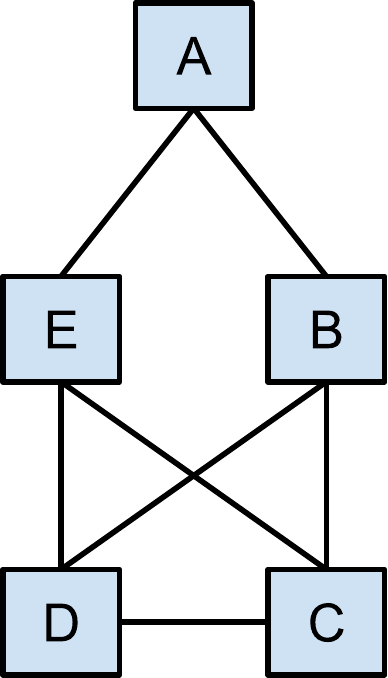}
\caption{Simple NANIP instance}
\end{subfigure}%
\begin{subfigure}[b]{0.6\textwidth}
\centering
\includegraphics[width=1.0\textwidth]{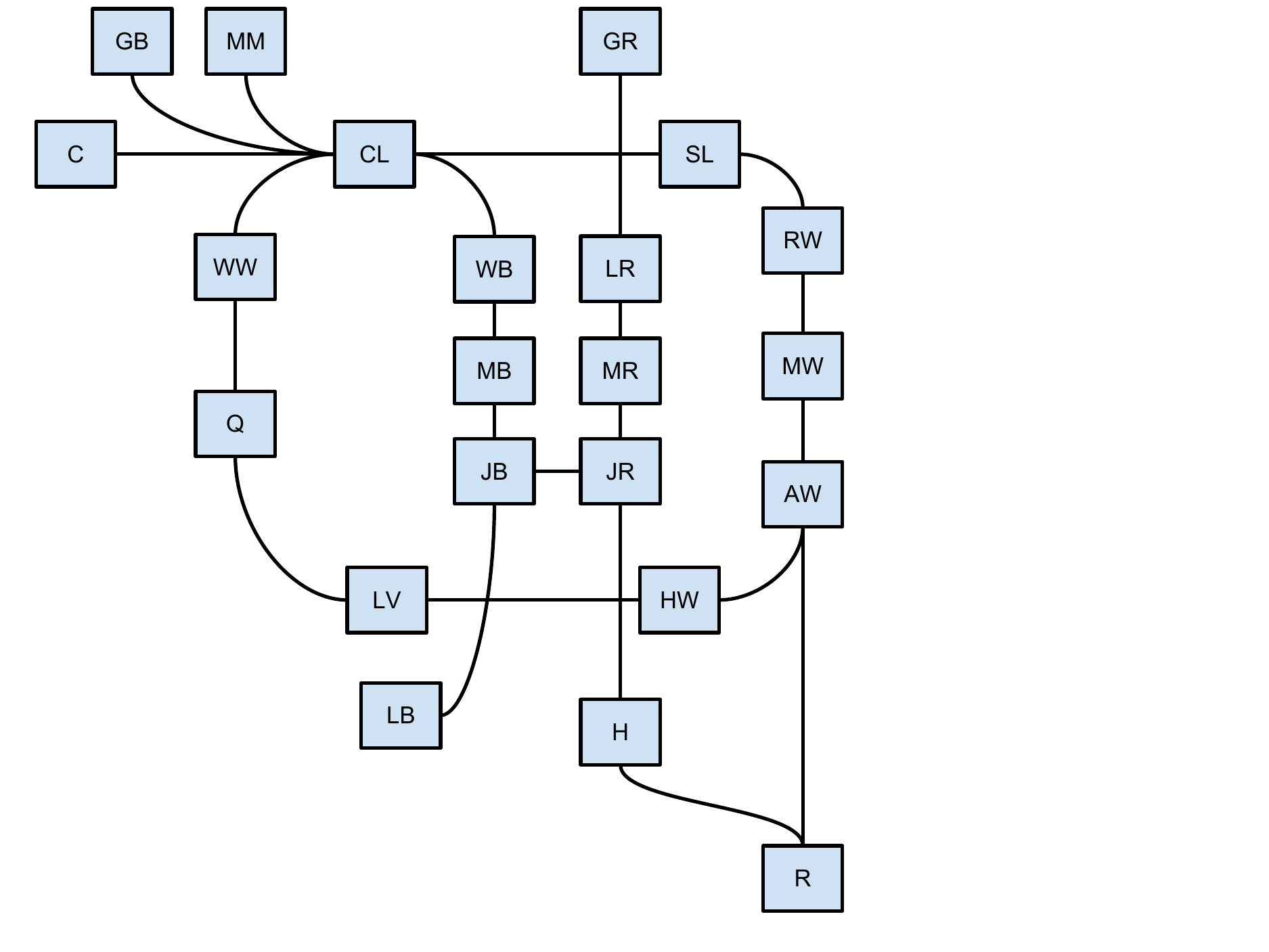} 
\caption{Central Chicago ``Loop''}
\end{subfigure}
\caption{Illustrations of NANIP.
(a) Simple instance.  
When $f(2)=2$, $f(1)=1$ and $f(k\geq2)=0$, the naive installation sequence $\sigma=(A,B,C,D,E)$
gives cost of $4=2+1+1+0+0$, but all optimal solutions have cost $3$.\label{fig:illustration}
(b) Actual metro stations and their connections in downtown Chicago ``Loop''.  
With the same $f$, any optimal sequences must recover Clark/Lake (CL) station before at least one of its neighbors.}
\end{figure}

We assume that $G$ is connected and undirected, unless we note otherwise. If
$G$ has multiple connected components, NANIP could be solved on each component
independently without affecting the total cost.

We begin by quoting a preliminary lemma from \cite{Gutfraind14} which
establishes that all the arguments used in calculating the node costs must sum
to $m$, the number of edges in the network.

\begin{lem}[\cite{Gutfraind14}]
\label{lem:edge-decomp}
For any network $G$, and any permutation $\sigma$ of the nodes of $G$, 
\begin{equation}
\sum_{t=1}^n r(v_t,G,\sigma) = m \label{eq:edge-decomp}\,.
\end{equation}
\end{lem}

One application of this lemma is the case of a linear cost function
$f(k)=ak+b$, for some real numbers $a$ and $b$.  With such a function the
optimization problem is trivial in that all installation permutations have the
same cost.

In the next section we will prove hardness results about NANIP; let us recall
some relevant definitions.

\begin{defn}

An optimization problem is called \emph{strongly NP-hard} if it is NP-hard and
the optimal value is a positive integer bounded by a polynomial of the input
size.

\end{defn}

\begin{defn}

An algorithm is an \emph{efficient polynomial time approximation scheme
(EPTAS)} for an optimization problem if, given a problem instance and an
approximation factor $\varepsilon$, it runs in time $O(F(\varepsilon) n^c)$
for some constant $c$ and some function $F$ and finds a solution whose
objective value is within an $\varepsilon$ fraction of the optimum. An EPTAS is
called a \emph{fully polynomial time approximation scheme (FPTAS)} it runs in
polynomial in the size of the problem instance and $\frac{1}{\varepsilon}$.  

\end{defn}

A strongly NP-hard optimization problem cannot have an FPTAS unless
P=NP: otherwise, if $n$ denotes the input size and $p$ denotes the polynomial
such that the optimum value is bounded by $p(n)$, setting
$\varepsilon=\frac{1}{2p(n)}$ for the FPTAS would yield an exact polynomial
time algorithm.

Some NP-hard problems become efficiently solvable if a natural parameter is
fixed to some constant. Such problems are called fixed parameter tractable.

\begin{defn}
FPT, the set of \emph{fixed parameter tractable} problems, is the set
of languages $L$ of the form $\langle x,k\rangle$
such that there is an algorithm running in time $O(F(k) n^c)$ for some function
$F$ and constant $c$ deciding whether $\langle x,k\rangle\in L$.
\end{defn}

An example of a fixed parameter tractable problem is the vertex cover problem
(where the parameter is the size of the vertex cover). 
Problems believed to be fixed parameter intractable include the graph coloring problem
(the parameter being the number of colors) and the clique problem (with the size of 
the clique as parameter).

For parametrized languages, there is a natural fixed parameter tractable analogue of polynomial time reductions.
These so-called \emph{fpt-reductions} are used to define hardness for classes of parametrized languages,
similarly to how NP-hardness is defined using polynomial time reductions.
One important class of parametrized languages is $W[1]$. For the definition of $W[1]$ and for more background
on parametrized complexity, we refer the reader to the monograph of Downey and Fellows~\cite{DowneyF13}.
They proved that under standard complexity-theoretic assumptions, $W[1]$ is a strict superset of
$FPT$; consequently, $W[1]$-hard problems are fixed parameter intractable.
We will use this fact to show the fixed parameter intractability of NANIP.

\section{Convex decreasing NANIP is NP-hard} \label{sec:computation} We now
consider the hardness of solving NANIP with convex decreasing cost functions.

\begin{thm} \label{thm:np-hard} The Neighbor Aided Network Installation Problem
is strongly NP-hard when $f$ is convex decreasing; as a consequence it admits
no FPTAS. \end{thm}

\begin{proof} 

We reduce from CLIQUE, that is, the problem of deciding given a graph $G =
(V,E)$ whether it contains as an induced subgraph the complete graph on $k$
vertices.  Given a graph $G = (V,E)$ with $n=|V|$ and an integer $k$, we
construct an instance of NANIP on a graph $G'$ with a convex cost function
$f(i)$ as follows.  Define $G'$ by adding $k$ new vertices $u_1, \dots, u_k$ to
$G$ which are made adjacent to every vertex in $V$ but not to each other, establishing
an independent set of size $k$.
Define the cost function 

\[
 f(i) = f_k(i) =
  \begin{cases} 
      \hfill k-i    \hfill & \text{ if $i \leq k$} \\
      \hfill 0          \hfill & \text{ otherwise} \\
  \end{cases}
\]

Let $M = \sum_{i=0}^{k} f(i)=\frac{k(k+1)}2$.  In a traversal $\sigma$ whose
first $k$ vertices yield cost $M$, every new vertex must be adjacent to every
previously visited vertex, i.e. the vertices form a $k$-clique.  Moreover, $M$
is the lower bound on the cost incurred by the first $k$ vertices of any
traversal of $G'$. 

Suppose that $G$ has a clique of size $k$, and denote by $v_1, \dots, v_k$ the
vertices of the clique, with $v_{k+1}, \dots, v_n$ the remaining vertices of
$G$. Then the following ordering is a traversal of $G'$ of cost exactly $M$:
\[
   v_1, \dots, v_k, u_1, \dots, u_k, v_{k+1}, \dots, v_n \,.
\]

Conversely, let $w_1, \dots, w_{n+k}$ be an ordering of the vertices of $G'$
achieving cost $M$.  Then by the above, the vertices $w_1, \dots, w_{k}$ must
form a $k$-clique in $G'$.  In the case these $k$ prefix vertices are all
vertices of $G$ we are done.  Otherwise, the independence of the $u_i$'s
implies that at most one $u_i$ is used in $w_1, \dots, w_{k+1}$; using more
would incur a total cost greater than $M$.  In this case the $k-1$ remaining
vertices of the prefix form a $(k-1)$-clique of $G$.
Since it is NP-hard to approximate CLIQUE within a polynomial factor
\cite{Zuckerman06}, this proves the NP-hardness of convex decreasing NANIP.

Moreover, since the optimum value of a NANIP instance obtained by this reduction
is at most $k^2$ which is upper bounded by $n^2$, the size of the NANIP instance,
it also follows that convex decreasing NANIP is strongly NP-hard and therefore
does not admit an FPTAS. 

\end{proof}

The cost function $f_k(i)$ used in the proof of Theorem~\ref{thm:np-hard} is
parametrized by $k$.  Call $\textup{NANIP}_k$ the subproblem of NANIP with cost
functions of finite support where the size of the support is $k$. Because we
consider $\textup{NANIP}_k$ a subproblem of general NANIP, stronger
parametrized hardness results for the former give insights about the latter.
Indeed, the following corollary is immediate.

\begin{cor}
$\textup{NANIP}_k$ is $W[1]$-hard.
\end{cor}

\begin{proof} 

CLIQUE is $W[1]$-complete when parametrized by the size of the clique.
$W[1]$-hardness is preserved by so-called $fpt$-reductions (see
\cite{DowneyF13}), and the reduction from the proof of
Theorem~\ref{thm:np-hard} is such a reduction. 

\end{proof}

In particular, standard complexity assumptions imply from this that
$\textup{NANIP}_k$ is not fixed-parameter tractable and has no efficient
polynomial-time approximation scheme (EPTAS). Now we will show that the same
reduction can be used to obtain a stronger approximation lower bound of $(1 +
n^{-c})$ for all $c > 0$. First a lemma.

\begin{lemma}
Let $G'$ and $f$ constructed as above, and let $\sigma$ denote a NANIP
traversal.  Suppose $V$ denote the vertices of $G$ and $U$ denote the vertices
of the independent set.  If $\sigma'$ is obtained from $\sigma$ by moving the
$U$ to positions $k+1,\ldots,2k$ (without changing the precedence relations of the vertices
in $V$), then $C_{G'}(\sigma')\le C_{G'}(\sigma)$.
\end{lemma}

\begin{proof}
Consider the positions in $\sigma$ of the first $k$ vertices from $G$, and let
$i_1, \dots, i_k$ be the positions of the vertices from $U$. Call $u_1 =
\sigma(i_1), \dots, u_k = \sigma(i_k)$.

\textbf{Case 1:} $i_1 > k$. In this case all the $u_i$ are free, as are all
vertices visited after $\sigma(k)$. If $i_1 > k+1$, apply the cyclic permutation
$\gamma_1 = (k+1, k+2, \dots, i_1)$ to move $u_1$ to position $k+1$. The cost
of visiting $u_1$ is still zero, and the cost of the other manipulated vertices
does not increase because they each gain one previously visited neighbor. Now
repeat this manipulation with $\gamma_s = (k+s, k+s+1, \dots, i_s)$ for $s = 2,
\dots, k$. An identical argument shows the cost never increases, and at the end
we have precisely $\sigma'$. 

\textbf{Case 2:} $i_1 \leq k$. In this case $u_1$ is not free. Let $j$ be the
index of the first $v \in V$ that occurs after $i_1$. Then apply the
cyclic permutation $\xi = (i_1, i_1 + 1, \dots, j)$ to move $v$ before $u_1$.
The cost of $v$ increases by at most $j - i_1$ (and this is not tight since it
is possible that $j > k+1$). But since all $\sigma(i_1), \sigma(i_1 + 1), \dots,
\sigma(j-1) \in U$, and they each gain a neighbor as a result of applying
$\xi$, so their total cost decreases by exactly $j - i_1$, and the total cost
of $\sigma$ does not increase. Now repeatedly apply $\xi$ (using the new values
of $i_1, j$) until $i_1 = k+1$. Then apply case 1 to finish.

\end{proof}

\begin{thm}
For all $c>0$, there is no efficient $(1+n^{-c})$-approximation
algorithm for NANIP on graphs with $n$ vertices with convex decreasing cost
functions, unless $\textup{P} = \textup{NP}$.  
\end{thm}

\begin{proof}
It is NP-hard to distinguish a clique number of at least $2^R$ from a clique
number of at most $2^{\delta R}$ in graphs on $2^{(1+\delta)R}$ vertices ($\delta>0$)
\cite{Zuckerman06}.  We will reduce this problem to finding an
$(1+n^{-c})$-approximation for NANIP.  In particular, we will show that there
is no efficient $C$-approximation approximation algorithm for NANIP, where $$ C
= \frac{k}{k+1} \left ( 1 + \frac{1}{k^{2\varepsilon}} \right ) $$ and
$k=n^{1/(1+\delta)}$.

This is equivalent to the statement of the theorem since by setting
$\eps=c/(2+2\delta)$, we get that there is no efficient
$\frac{n^{1+\delta}}{n^{1+\delta}+1}(1+n^{-c})<(1+n^{-c})$-approximation
algorithm for NANIP.

Let $G$ be a graph on $n=2^{(1+\delta)R}$ vertices containing a $k$-clique
where $k=n^{1/(1+\delta)}=2^R$ and construct $G'$ from $G$ by adding a
$k$-independent set as before, with $f(i)=\max(k-i, 0)$. Suppose we have an
efficient $C$-approximation algorithm for NANIP. After running it on input
$(G', k)$, modify the output sequence according to the previous lemma. Then all
the nodes after the first $k$ are free, thus the cost of the sequence is
determined by the first $k$ vertices. Since they all have fewer than $k$
preceding neighbors, the cost function for them is linear, implying that the
total cost of the sequence depends only on the number of edges in between the
first $k$ vertices.

The cost of the optimal NANIP sequence in $G'$ is $k(k+1)/2$, thus the cost of
the sequence returned by the approximation algorithm is at most

$$ \frac{k}{k+1}\left(1+\frac{1}{k^{2\eps}}\right)\cdot \frac{k(k+1)}{2} =
\frac12(k^2+k^{2-2\eps}).$$

Since $$\frac12(k^2+k^{2-2\eps})=(-1)(1-k^{-2\eps})\frac{k^2}{2}+k^2,$$ it
follows by \cite{Gutfraind14}, Corollary 2, that there are more than
$(1-k^{-2\eps})k^2/2$ edges between the first $k$ vertices.

Tur\'an's theorem~\cite{Turan1941} states that, a graph on $k$ vertices that
does not contain an $(r+1)$-clique can have at most $(1-\frac1r)k^2/2$ edges.
The contrapositive implies that the induced subgraph on the first $k$ vertices
of the NANIP sequence contains a $(k^{2\eps}-1)$-clique.  Since $k^{2\eps}-1>
2^{\eps R}$, this completes the proof.

\end{proof}

\section{Greedy analysis for convex NANIP}

In this section we discuss the approximation guarantees of the greedy algorithm
on convex NANIP. The greedy algorithm is defined to choose the
cheapest cost vertex at every step, breaking ties arbitrarily. A useful
observation here is that the greedy algorithm always produces a connected
traversal of a connected graph, in the sense that every prefix of the final
traversal induces a connected subgraph. We call an algorithm which always produces a
connected traversal a \emph{connected algorithm}.

Our next theorem shows a rather surprising result, that optimal recovery
sometimes requires disconnected solutions, even on convex cost functions.
Connected solutions can perform quite badly, having a cost that is a
$\Omega(\log n)$ multiple of the optimum.

\begin{thm}
Connected algorithms have an approximation ratio $\Omega(\log(n))$ for convex
NANIP problems.
\end{thm}

\begin{proof}

We construct a particular instance for which a connected algorithm incurs cost
$\Omega(\log(n))$ while the optimal route has constant cost. Define the graph
$B(m)$ to be a complete binary tree $T$ with $m$ levels, and a pair of vertices
$u,v$ such that the leaves of $T$ and $\{u,v\}$ form the complete bipartite
graph $K_{2^{m-1}, 2}$. As an example, $B(3)$ is given in Figure~\ref{fig:b3}.

\begin{figure}[th]
\centering
\begin{subfigure}{.5\textwidth}
  \centering
  \scalebox{0.45}{\includegraphics{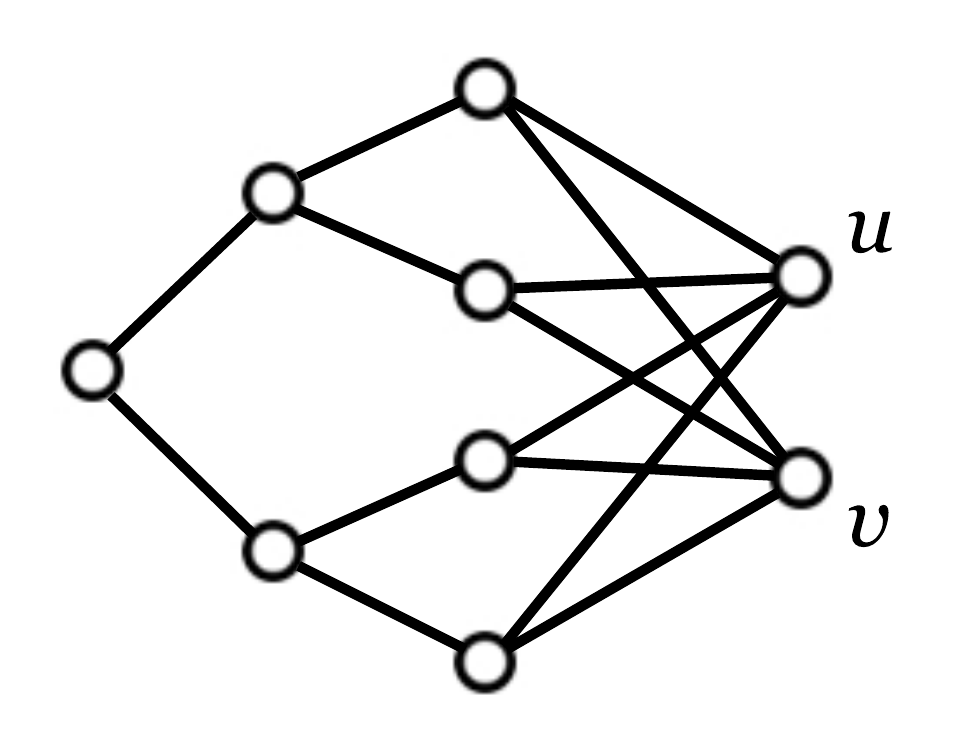}}
\end{subfigure}%
\begin{subfigure}{.5\textwidth}
  \centering
  \scalebox{0.40}{\includegraphics{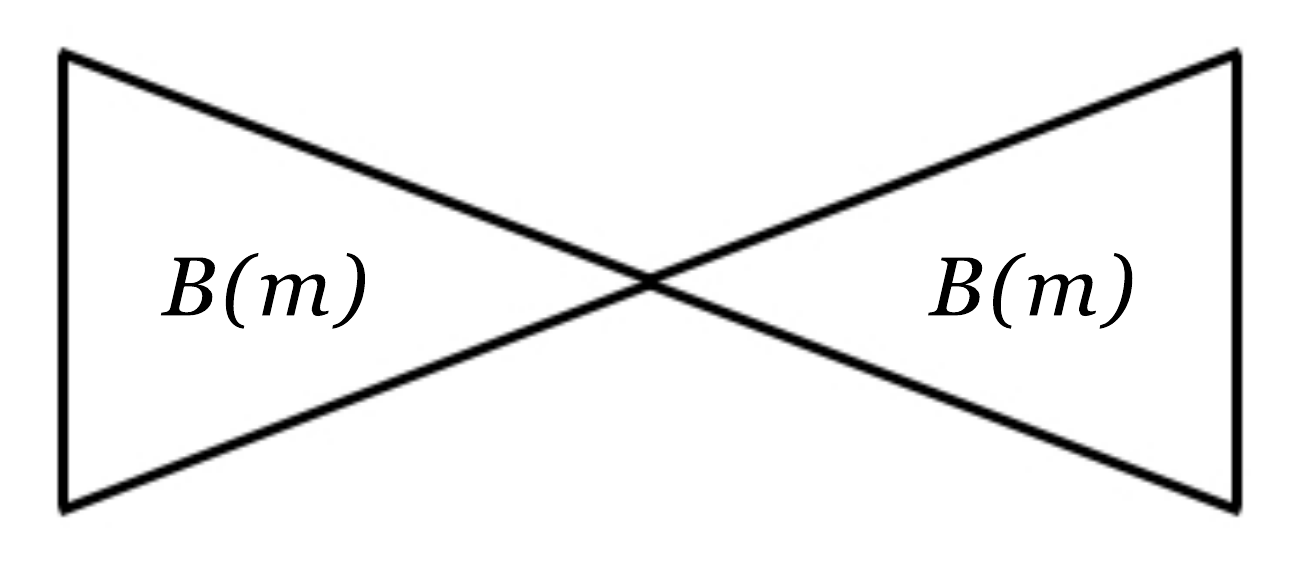}}
\end{subfigure}
\caption{Left: the graph $B(3)$; Right: two $B(m)$ pieced together to force a
connected algorithm to incur $\Omega(\log(n))$ cost.}
\label{fig:b3}
\end{figure}

Define the cost function $f(n)$ such that $f(0) = 2, f(1) = 1$, and $f(n) = 0$
for all $n \geq 2$. For this cost function it is clear that the minimum cost of
a traversal of $B(m)$ is exactly 4 by first choosing the two vertices of $B(m)$
that are not part of the tree, and then traversing the rest of the tree at zero
cost. However, if a connected algorithm were forced to start at the root of the
tree, it would incur cost $\Omega(m) = \Omega(\log(n))$ since every vertex
would have at most one visited neighbor. 

To force such an algorithm into this situation we glue two copies of $B(m)$
together so that their trees share a root. Then any connected ordering must
start in one of the two copies, and to visit the other copy it must pass
through the root, incurring a total cost of $\Omega(\log(n))$. On the other
hand, the optimal traversal has total cost 8. 

\end{proof}

Further, the greedy algorithm, which simply chooses the cheapest vertex at each
step and breaks ties arbitrarily, gives a $\Theta(n)$ approximation ratio in
the worst case. To see this, note that in the construction from the theorem the
only way a connected algorithm can achieve the logarithmic lower bound is by
traveling directly from the root to the leaves. But by breaking ties
arbitrarily, the greedy algorithm may visit every interior node in the tree
before reaching the leaves, thus incurring a linear cost overall.

\section{Integer programming for NANIP}\label{section:IP}

In this section we describe a new integer programming (IP) formulation of the
NANIP problem by adding in Miller-Tucker-Zemlin-type subtour elimination
constraints~\cite{miller1960integer}.  An IP, of course, does not give a
polynomial time algorithm, but can be sufficiently fast for some instances of
practical interest.  We then show that this formulation, experimentally,
improves on the previous formulation by \cite{Gutfraind14}.

\subsection{A new integer program}

In what follows we will assume that the cost function $f$ is a continuous
convex decreasing function $\mathbb{R}^{\geq 0} \to \mathbb{R}^{\geq 0}$ rather
than one $\mathbb{N} \to \mathbb{R}^{\geq 0}$. It is necessary to extend $f$ to
a continuous function for the LP relaxation to be well-defined.  While there
are many ways to do so, formulating the IP for a general continuous $f$
encapsulates all of them.

For an undirected graph $G = (V,E)$ on $n = |V|$ vertices, and introduce the
arc set $A$ by replacing each undirected edge with two directed arcs.  For all
$(i,j)\in A$ define variables $e_{ij} \in \{ 0,1 \}$.  The choice $e_{ij} = 1$
has the interpretation that $i$ is traversed before $j$ in a candidate ordering
of the vertices, or that one chooses the directed edges $(i,j)$ and discards
the other. In order to maintain consistency of the IP we impose the constraint
$e_{ij} = 1 - e_{ji}$ for all edges $(i,j)$ with $i < j$. Finally, we wish to
enforce that choosing values for the $e_{ij}$ corresponds to defining a partial
order on $V$ (i.e., that the subgraph of chosen edges forms a DAG). We use the
subtour elimination technique of Miller, Tucker, and Zemlin
\cite{miller1960integer} and introduce variables $u_i$ for $i = 1, \dots, n$
with the constraints

\begin{align}
\label{eq:dag-constraint}
\begin{matrix}
   u_i - u_j + 1 \leq n (1 - e_{ij}) & \forall (i,j) \in A \\ 
   0 \leq u_i \leq n & i = 1, \dots, n 
\end{matrix}
\end{align}
\noindent Thus, if $i$ is visited before $j$ then $u_i \geq u_j - 1$.
Now denote by $d_i = \sum_{(j,i) \in E} e_{ji}$, which is the number of
neighbors of $v_i$ visited before $v_i$ in a candidate ordering of $V$. The
objective function is the convex function $\sum_{i} f(d_i)$, and putting these
together we have the following convex integer program: 

\begin{figure}[th]
\begin{centering}
\begin{align*}
\textup{min }  & \sum_i f(d_i)                       & \\ 
\textup{s.t. } & d_i = \sum_{(j,i) \in A} e_{ji}     & i = 1, \dots, n \\ 
               & e_{ij} = 1 - e_{ji}                 & (i,j) \in A, i < j \\ 
               & u_i - u_j + 1 \leq n (1 - e_{ij})   & (i,j) \in A \\ 
               & 0 \leq u_i \leq n                   & i = 1, \dots, n \\ 
               & e_{ij} \in \{0,1\}                    & (i,j) \in A
\end{align*}
\end{centering}
\end{figure}

The integer program has a natural LP relaxation by replacing the integrality
constraints with $0 \leq e_{ij} \leq 1$. 
Because $f$ is only evaluated at integer points, it is possible to replace $f(d_i)$
with a real-valued variable bound by a set of linear inequalities, as detailed in~\cite{Gutfraind14}.

\subsection{Experimental results}

We compared the new IP formulation with the formulation of
\cite{Gutfraind14}, in the algebraic optimization framework
(IBM ILOG CPLEX 12.4 solver) running with a single thread on Intel(R) Core(TM)
i5 CPU U 520  @ 1.07GHz with 3.84E6 kB of random access memory.  The simulation
used graphs on 15 nodes, where the number of edges was increased from 14 (tree)
until the running time exceeded 1 hour.  For each edge density, we constructed
5 graphs and reported the average running time of the two algorithms.

From the computational experiments it is clear that our formulation gives
significant improvements.  For instance, the solve time seems to not depend on
the number of nodes in the graph (Fig.~\ref{fig:iptime}(a)), unlike in the
previous formulation.  We are also able to solve NANIP instances on 45 edges in
under an hour, whereas the previous formulation solved only 30 edge graphs in
that time (Fig.~\ref{fig:iptime}(b)).

\begin{figure}[th]
\begin{centering}
(a)\includegraphics[width=0.45\textwidth]{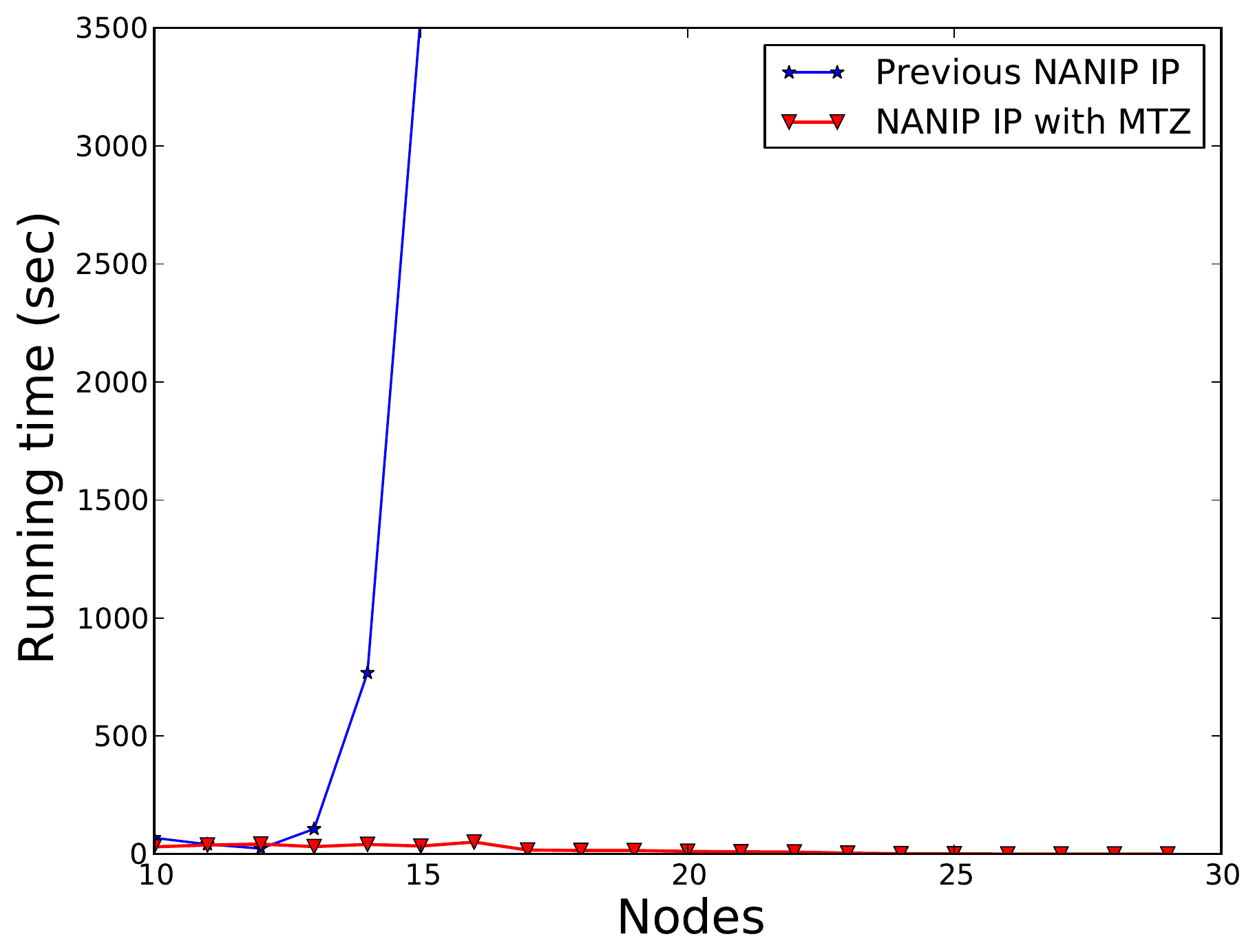} 
(b)\includegraphics[width=0.45\textwidth]{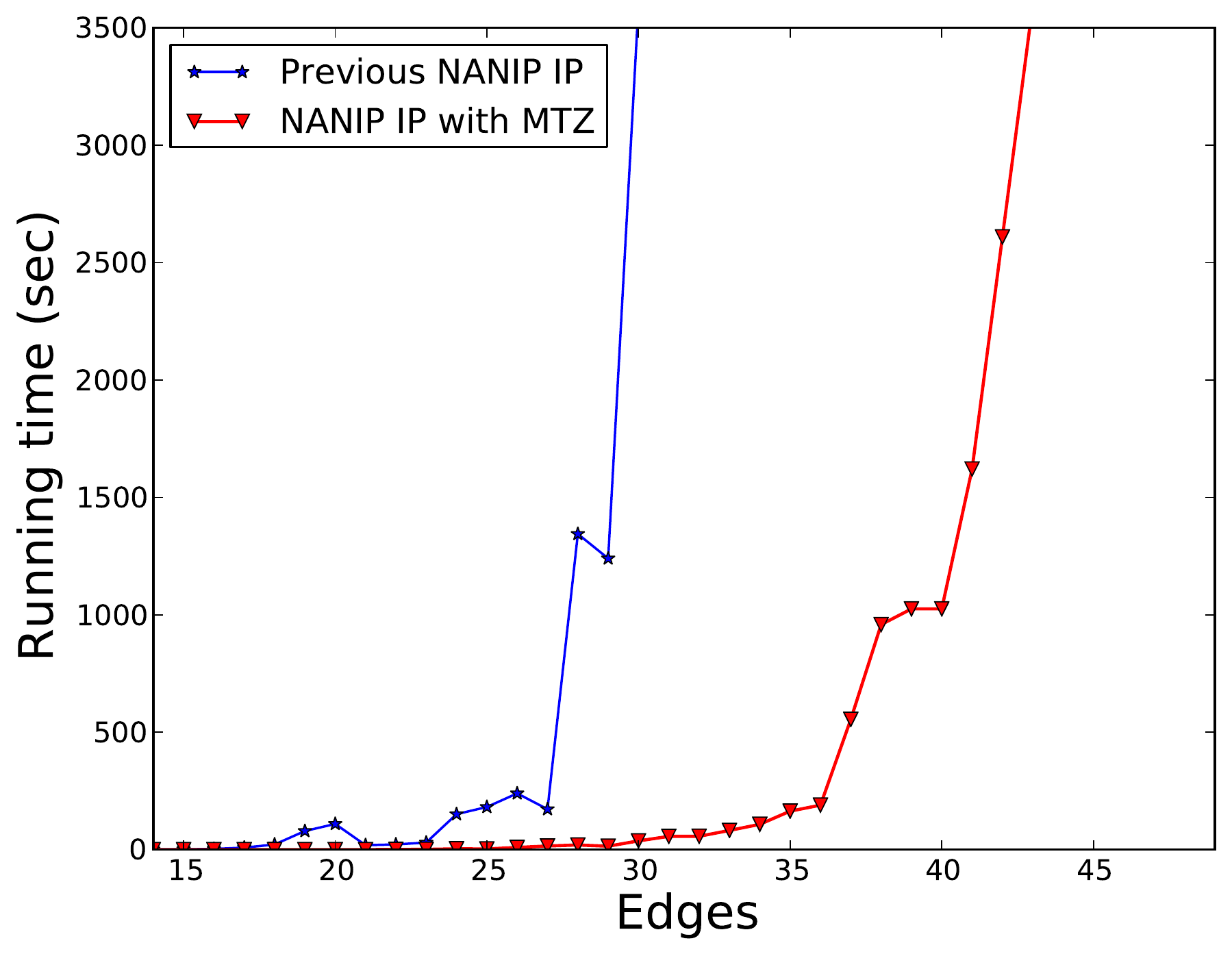} 
\par
\end{centering}
\caption{A comparison of the formulations in
\cite{Gutfraind14} and our new IP formulation with MTZ-type
constraints.  This graph plots running time vs.\ (a) number of nodes and, (b)
number of edges in the target graph.
In (a) the number of edges was kept at 30 throughout, while in (b) the number
of nodes was 15 throughout.\label{fig:iptime}} 
\end{figure}

\section{Conclusion} \label{sec:concl}
We analyzed the recently introduced Neighbor-Aided Network Installation
Problem.  We proved the NP-hardness of the problem for the practically most
relevant case of convex decreasing cost functions, addressing an open problem
raised in \cite{Gutfraind14}.  We then showed that the worst case approximation
ratio of the natural greedy algorithm is $\Theta(n)$.  We also gave a new IP
formulation for optimally solving NANIP, which outperforms previous
formulations.

The approximability of NANIP remains an open problem. In particular, it is
still not known whether an efficient $o(n)$ approximation algorithm exists for
general convex decreasing cost functions. One obstacle to finding a good
rounding algorithm is that the IP we presented has an infinite integrality gap.
As proof, the graph $K_n$ with the function $f(i) = \max(0, n/2 - i)$ has
$\textup{OPT} = \Omega(n^2)$ but the linear relaxation has $\textup{OPT}_{LP} =
0$. So an approximation algorithm via LP rounding would require a different IP
formulation.

\section*{Acknowledgments and Funding}
We thank our colleagues for insightful discussions.
AG was supported in part by an ORISE fellowship at the Food and Drug Administration.
CPLEX software was provided by IBM through the IBM Academic Initiative program.

\bibliographystyle{plain}
\bibliography{nanip}

\begin{thebibliography}{10}

\bibitem{Adibi94}
M.M. Adibi and L.H. Fink.
\newblock Power system restoration planning.
\newblock {\em {IEEE} Transactions on Systems}, 9(1):22 --28, 1994.

\bibitem{Bertoli02}
P.~Bertoli, R.~Cimatti, J.~Slaney, and S.~Thibaux.
\newblock Solving power supply restoration problems with planning via symbolic
  model checking.
\newblock In {\em AIPS-02 Workshop on Planning via Model-Checking}, pages
  576--580, 2002.

\bibitem{coffrin2011strategic}
Carleton Coffrin, Pascal Van~Hentenryck, and Russell Bent.
\newblock Strategic stockpiling of power system supplies for disaster recovery.
\newblock In {\em Power and Energy Society General Meeting, 2011 IEEE}, pages
  1--8. IEEE, 2011.

\bibitem{DowneyF13}
Rodney~G. Downey and Michael~R. Fellows.
\newblock {\em Fundamentals of Parameterized Complexity}.
\newblock Texts in Computer Science. Springer, 2013.

\bibitem{Guha99}
Sudipto Guha, Anna Moss, Joseph~(Seffi) Naor, and Baruch Schieber.
\newblock Efficient recovery from power outage (extended abstract).
\newblock In {\em Proceedings of the thirty-first annual ACM symposium on
  Theory of computing}, STOC '99, pages 574--582, New York, NY, USA, 1999. ACM.

\bibitem{Gutfraind14}
Alexander Gutfraind, Milan Bradonji\'c, and Tim Novikoff.
\newblock Modelling the neighbour aid phenomenon for installing costly complex
  networks.
\newblock {\em Journal of Complex Networks}, 2014.
\newblock doi:10.1093/comnet/cnu033.

\bibitem{Karp61}
Michael Held and Richard~M. Karp.
\newblock A dynamic programming approach to sequencing problems.
\newblock In {\em ACM '61: Proceedings of the 1961 16th ACM national meeting},
  pages 71.201--71.204, New York, NY, USA, 1961. ACM.

\bibitem{Lee07}
E.E. Lee, J.E. Mitchell, and W.A. Wallace.
\newblock Restoration of services in interdependent infrastructure systems: A
  network flows approach.
\newblock {\em Systems, Man, and Cybernetics, Part C: Applications and Reviews,
  IEEE Transactions on}, 37(6):1303 --1317, Nov 2007.

\bibitem{miller1960integer}
Clair~E Miller, Albert~W Tucker, and Richard~A Zemlin.
\newblock Integer programming formulation of traveling salesman problems.
\newblock {\em Journal of the ACM (JACM)}, 7(4):326--329, 1960.

\bibitem{Mitchell96}
J.E. Mitchell and B.~Borchers.
\newblock Solving real-world linear ordering problems using a primal-dual
  interior point cutting plane method.
\newblock {\em Annals of Operations Research}, 62(1):253--276, 1996.

\bibitem{nurre2010restoring}
Sarah~G Nurre and TC~Sharkey.
\newblock Restoring infrastructure systems: An integrated network design and
  scheduling problem.
\newblock In {\em Proceedings of the 2010 Industrial Engineering Research
  Conference}, 2010.

\bibitem{schrijver2005history}
A.~Schrijver.
\newblock {On the history of combinatorial optimization (till 1960)}.
\newblock {\em Handbooks in Operations Research and Management Science},
  12:1--68, 2005.

\bibitem{Turan1941}
Paul Tur\'an.
\newblock On an extremal problem in graph theory.
\newblock {\em Matematikai \'es Fizikai Lapok}, 48:436--452, 1941.

\bibitem{Hentenryck10}
P.~Van~Hentenryck, R.~Bent, and C.~Coffrin.
\newblock Strategic planning for disaster recovery with stochastic last mile
  distribution.
\newblock In Andrea Lodi, Michela Milano, and Paolo Toth, editors, {\em
  Integration of AI and OR Techniques in Constraint Programming for
  Combinatorial Optimization Problems}, volume 6140 of {\em LNCS}, pages
  318--333. Springer Berlin / Heidelberg, 2010.

\bibitem{West01}
Douglas~B. West.
\newblock {\em Introduction to Graph Theory}.
\newblock Pearson Prentice Hall, New Jersey, 2001.

\bibitem{Zuckerman06}
David Zuckerman.
\newblock Linear degree extractors and the inapproximability of max clique and
  chromatic number.
\newblock In {\em Proceedings of the Thirty-eighth Annual ACM Symposium on
  Theory of Computing}, STOC '06, pages 681--690, New York, NY, USA, 2006. ACM.

\end{thebibliography}

\end{document}